\documentclass[preprint, 12pt, a4paper]{elsarticle}

\usepackage{graphicx}
\usepackage{tikz}
\usetikzlibrary{arrows}

\newdefinition{definition}{Definition}
\newdefinition{example}{Example}
\newdefinition{remark}{Remark}
\newtheorem{lemma}{Lemma}
\newtheorem{theorem}{Theorem}
\newproof{proof}{Proof}

\usepackage{soul}
\usepackage{stmaryrd}
\usepackage{amssymb}
\usepackage{amsmath}
\usepackage{amsfonts}

\newcommand{\N}{\mathbb{N}}
\newcommand{\TES}[1]{\mathit{TES(#1)}}

\newcommand{\Po}{\mathcal{P}}
\newcommand{\Rp}{\mathbb{R}_+}

\newcommand{\pr}{\mathrm{pr}}

\begin{document}
\begin{frontmatter}
\title{A Formal Framework for Distributed Cyber-Physical Systems}
\author[1]{Benjamin Lion}
\author[2]{Farhad Arbab}
\author[3]{Carolyn Talcott}
\address[1]{Leiden University, Leiden, The Netherlands}
\address[2]{CWI, Amsterdam, The Netherlands}
\address[3]{SRI International, CA, USA}
\begin{abstract}
    Composition is an important feature of a specification language, as it enables the design of a complex system in terms of a product of its parts.
    Decomposition is equally important in order to reason about structural properties of a system.
    Usually, however, a system can be decomposed in more than one way, each optimizing for a different set of criteria. 
    We extend an algebraic component-based model for cyber-physical systems to reason about decomposition.
    In this model, components compose using a family of algebraic products, and decompose, under some conditions, given a corresponding family of division operators.
    We use division to specify invariants of a system of components, and to model desirable updates. 
    We apply our framework to design a cyber-physical system consisting of robots moving on a shared field, and identify desirable updates using our division operator.
\end{abstract}
\end{frontmatter}

\section{Introduction}
We present a design framework for cyber-physical systems that enables reasoning about composition and decomposition.
It is common, when modeling a system, to separate its design from its implementation. 
A design framework employs high level primitives to simplify the specification of \emph{what} behavior is desirable. 
A formal design, moreover, enables high level analysis that benefits later implementation (e.g., proof of the existence of an implementation, verification of safety properties, etc.). 
In contrast, an implementation uses low level operations to specify \emph{how} a behavior is constructed.
An implementation provides a precise description of the system's behavior that can be tested, simulated, and under some conditions, proved correct with respect to the design specification.

Composition is the act of assembling components to form complex systems.
A formal model's ability to support such construction is a particularly desirable property if the underlying parts have different types of specifications (e.g., continuous versus discrete), but still need to communicate and interact.
Our work in~\cite{DBLP:journals/corr/abs-2110-02214} presents a component model that captures both discrete and continuous changes, for which timed-event sequences (TESs) form instances of a component behavior. An observation is a set of events with a unique time stamp. A component has an interface that defines which events are observable, and a behavior that denotes all possible sequences of its observations (i.e., a set of TESs).  The precise machinery that generates such component is abstracted away. Instead, we present interaction between components as an algebra on components, that includes a wide variety of user defined operations. 

Decomposition is dual to composition, as it simplifies a component behavior by removing some of its parts. 
Decomposition is interesting in two ways: it gives insight on whether a system is a composite that includes a specific component, and it returns a subsystem that, in composition with that component, would give back the initial system.
Decomposition is not unique, and may induce a ``cost'' or ``measure'' on its alternatives, i.e., a component $A$ may be seen as a product $B \times C$ or $B \times D$ with $C \not = D$, where by some measures, the cost of using $C$ or $D$ may differ.
While the observable behavior may not change, i.e., the set of sequences of observations stays the same, the substitution of a component with another may somehow \emph{improve} the overall system, e.g., by enhancing its efficiency. For instance, running time is often omitted when specifying systems whose behaviors are oblivious to time itself. However, in practice, the time that a program takes to process its inputs matters. Thus, a component may be substituted with a component exposing the same behavior but running faster. 
Other criteria such as the size of the implementation, the cost of the production, procurement, maintenance, etc., may be considered in changing one component for another.
In this paper, we also consider an orthogonal concern: the cost of \emph{coordination}.
Intuitively, the cost of coordination captures the fact that events of two components may be tightly related. For example, if two events are related, the occurrence of an event in one component implies the occurrence of some events in another component. While such constraints are declarative in our model, their implementation may be costly. Thus, the relation between observable events of two components may increase the underlying cost to concurrently execute those two components. 
Finally, having an operation to study system decomposition brings an alternative perspective on fault detection and diagnosis~\cite{DBLP:journals/scp/KappeLAT19}.

Formally, we extend our algebra of components~\cite{DBLP:journals/corr/abs-2110-02214} with a new type of operator: a division operation.
Division intuitively models decomposition, and acts as an inverse composition operation.
Practically, the division of component $A$ by component $B$ returns one component $C$ from all components $D$ such that $A = B \times D$.
Different cost models give rise to different operations of division by picking specific alternatives for $D$. 
We abstractly reason about cost using a partially ordered set of components and show that, for some orders, the set of candidates naturally gives rise to a maximal (minimal) element.

As a running example, we consider a set of robots moving continuously on a shared field.
We use the operation of division to specify desirable updates that would prevent robots from interfering with other robots. We also apply division to find simpler components that, if used, would still preserve the entire system behavior. We finally specify the necessary coordination for the robots to self sort on the field.
\paragraph{We summarize our contributions in this paper as:}
\begin{itemize}
    \item an extension of the algebra of components with a division operator;
    \item a result on natural ordering of components; 
    \item application of division for specifying valid updates of components.
\end{itemize}

\section{Preliminaries}
A Timed-Event Stream (TES) $\sigma$ over a set of events $E$ is an infinite sequence of \emph{observations}, where its $i^{th}$ observation $\sigma(i) = (O, t)$, $i \in \N$, consists of a pair of a set of events $O \subseteq E$, called the \emph{observable}, and a positive real number $t \in \Rp$ as time stamp.
A TES has the additional properties that its consecutive time stamps are monotonically increasing and non-Zeno, i.e., if $\sigma(i) = (O_i, t_i)$ is the $i^{th}$ element of TES $\sigma$, then (1) $t_i < t_{i+1}$, and (2) for any time $t \in \Rp$, there exists an element $\sigma(i) = (O_i, t_i)$ in $\sigma$ such that $t< t_i$.
We use $\sigma'$ to denote the derivative of the stream $\sigma$, such that $\sigma'(i) = \sigma(i+1)$ for all $i \in \N$.
We refer to the stream of observables of $\sigma$ as its first projection $\pr_1(\sigma) \in \Po(E)^\omega$, and the stream of time stamps as its second projection $\pr_2(\sigma) \in \Rp^\omega$, such that $\sigma(i) = (\pr_1(\sigma)(i), \pr_2(\sigma)(i))$ for all $i \in \N$.
We write $\sigma(t) = O$ if there exists $i \in \N$ such that $\sigma(i) = (O,t)$, and $\sigma(t) = \emptyset$ otherwise. We use $\it{dom}(\sigma)$ to refer to the set of observable time stamps, i.e., the set $\it{dom}(\sigma) = \{ t\in\Rp \mid \exists i. \pr_2(\sigma)(i)=t\}$.
\subsection{Components}
A component in our model does not include the detail of \emph{how} elements of its behavior are constructed, but specifies only \emph{what} sequences of observations over time are acceptable.
Practically, there might be different ways to construct the same parts of a component behavior, but we abstract away the details of such implementation as unobservable.
We use $\TES{E}$ to denote the set of all TESs over the set of events $E$.
\label{sec:component}

\begin{definition}[Component]
     \label{def:component}
     A \emph{component} $C = (E,L)$ is a pair of a set of events $E$, called its \emph{interface}, and a \emph{behavior} $L \subseteq \TES{E}$.
 \end{definition}
 Given component $A = (E_A, L_A)$, we write $\sigma:A$ for a TES $\sigma \in L_A$.
In order to demonstrate the applicability of our framework, we model a robot, a field, and a protocol as components.
We show how components capture both cyber and physical aspects of systems.

 \newcommand{\Stop}{\mathit{stop}}
 \newcommand{\rread}{\mathit{read}}
\begin{example}[Roaming robots]
    \label{ex:robots}
    We capture, as a component, sequences of observations emerging from discrete actions of a robot at a fixed time frequency.
    For simplicity, we consider that the robot can perform actions of two types only: a move in a cardinal direction, and a read of its position sensor.
    A move action of robot $i$ creates an event of the form $\it{d(i,p)}$ where $d$ is the direction, and $p$ is the power required for the move.
    The read action of robot $i$ generates an event of the form $\it{read(i,(x;y))}$ where $(x;y)$ is a coordinate location.

    Formally, we write $R(i,T,P) = (E_R(i,P), L_R(T))$ for the robot component with identifier $i$ with $E_R(i,P)$ the set
    \[\{ S(i,p),W(i,p),N(i,p),E(i,p), \rread(i,(x;y)) \mid x,y\in \llbracket -20,20\rrbracket, p \leq P\}\]
    and $L_R(T) \subseteq \TES{E_R(i,P)}$ be such that all observations are time stamped with a multiple of the period $T \in \Rp$, i.e., for all $\sigma \in L_R(T)$, if $(O,t) \in \sigma$ then there exists $k\in\N$ such that $t = k\cdot T$.
    The component $R(i,T)$ therefore captures all robots whose directions are restricted to S(outh), W(est), N(orth), and E(ast), whose power is limited to $P$, and whose location values are integers in the interval $\llbracket-20,20\rrbracket$. The robot stops whenever $p = 0$.

    In Table~\ref{table:pref}, we display the prefix of one TES from the behavior of three robot components. Note that each line corresponds to a time instant, for which each robot may or may not have observed events. The symbol `$-$' represents no observable, while otherwise we show the set of events observed.
    The time column is factorized by the period $T$, shared by all robots.
    Thus, at time $3\cdot T$, robot $R(1,T,P)$ moves west, while robot $R(2,T,P)$ moves north, and robot $R(3,T,P)$ moves east, all with required power $p$.
\begin{table}
\centering
\caption{Three prefixes of timed-event sequences for $R(1,T,P)$, $R(2,T,P)$, and $R(3,T,P)$, where $T$ and $P$ are fixed, and each move action consumes the same power $p$.}
\begin{tabular}{cccc}
     $t/T$  &\ $\sigma:R(1,T,P)$\ &\ $\tau:R(2,T,P)$\ &\ $\delta:R(3,T,P)$ \\
    \hline
    $1$ & $\{N(1,p)\}$ & $-$ & $-$  \\
    $2$ & $\{W(1,p)\}$ & $-$& $-$ \\
    $3$ & $\{W(1,p)\}$  & $\{N(2,p)\}$ & $\{E(3,p)\}$ \\
    $4$ & $\{S(1,p)\}$ & $\{W(2,p)\}$       & $\{E(3,p)\}$ \\
    $...$ & $...$ & $...$          & $-$
\end{tabular}
\label{table:pref}
\end{table}
 \hfill$\blacksquare$
\end{example}

In the robot component of Example~\ref{ex:robots}, observations occur at a fixed frequency.
For some physical components, however, observations may occur at any point in time.
For instance, consider the field on which the robot moves. Each time a robot moves may induce a change in the field's state, and the field's state may be observable at any time or frequency. Internally, the field may record its state changes by a continuous function, while restricting the possibilities of the robots to move due to physical limitations.
We describe the field on which the robot moves as a component, and we specify, in Example~\ref{ex:fr-sign}, how robot components interact with the field component. 

\begin{example}[Field]
\label{ex:grid}
The field component captures, in its behavior, the dynamics of its state as a sequence of observations.
The collection of objects on a field is given by the set $I$.
The state of a field is a triple $(((x;y)_i)_{i\in I}, (\overrightarrow{v_i})_{i\in I}, t)$ that describes, at time $t \in \Rp$, the position $(x;y)_i$ and the velocity $\overrightarrow{v_i}$ of each object in $I$.
We model each object in $I$ by a square of dimension $1$ by $1$, and the coordinate $(x;y)_i$ represents the central position of the square.
We use $\mu= (((x_0,y_0)_i)_{i\in I}, (\overrightarrow{v_{0i}})_{i\in I}, t_0)$ as an initial state for the field, which gives for each robot in $i\in I$ a position and an initial velocity. 
Note that static obstacles on the field can be modeled as objects $i\in I$ with position $(x;y)_i$ and zero velocity.

Formally, the field component is the pair $F_\mu(I) = (E_F(I), L_F(I,\mu))$ with 
\[E_F(I) = \{ (x;y)_i, \it{move(i,\overrightarrow{v})}\mid i \in I,\ x,y\in \mathbb{R},\ \overrightarrow{v} \in \mathbb{R}\times\mathbb{R}  \}\]
where each event $\it{move(i, \overrightarrow{v})}$ continuously moves object $i$ with velocity $\overrightarrow{v}$, and event $(x;y)_i$ displays the location of object $i$ on the field.

The set $L_F(I,\mu) \subseteq \TES{E_F(I)}$ captures all sequences of observations that consistently sample trajectories of each objects in $I$, according to the change of state of the field and the internal constraint.
As a physical constraint, we impose that no two objects can overlap, i.e., for any disjoint $i,j \in I$ and for all time $t\in \Rp$, with $(x;y)_i$ and $(u;v)_j$ their respective positions, then $[x-0.5,x+0.5] \cap [u-0.5,u+0.5] = [y-0.5,y+0.5] \cap [v-0.5,v+0.5] = \emptyset$.
Even though the mechanism for such a constraint is hidden in the field component, typically, the move of a robot is eventually limited by the physics of the field.
\hfill $\blacksquare$
\end{example}

\begin{remark}
    There is a fundamental difference between the robot component in Example~\ref{ex:robots} and the field component in Example~\ref{ex:grid}. The robot component has an adequate underlying sampling frequency which prevents missing any event if observations are made by that frequency. 
    However, the field has no such frequency for its observations, which means that there may be another intermediate observation occurring between any two observations. In~\cite{DBLP:journals/corr/abs-2110-02214}, we capture,  as a behavioral property, the property for a component to interleave observations between any two observation.
\end{remark}

\newcommand{\swap}{\mathit{swap}}
\begin{example}[Protocol]
    \label{ex:protocol}
As shown in Example~\ref{ex:grid}, physics may impose some constraints that force robots to coordinate.
A protocol is a component that, for instance, coordinates the synchronous movement of a pair of robots.
For example, when two robots face each other on the field, the $\swap(i,j)$ protocol moves robot $R(i,P,T)$ north, west, and then south, as it moves robot $R(j,P,T)$ east.
Note that the protocol requires the completion of a sequence of moves to succeed. Another robot could be in the way, and therefore delay the last observables of the sequence.
The swap component is defined by $\swap(i,j) = (E_P(i,j), L_P(i,j))$ where $E_P(i,j) = E_R(i,P) \cup E_R(j,P)$ and $L_P(i,j)$ captures all sequences of observations where the two robots $i$ and $j$ swap positions.
\hfill $\blacksquare$
\end{example}

Components are declarative entities that may denote either the behavior of a specification, or the behavior of an implementation.
The usual relation between the behavior of a program and the property of such program constitutes a refinement relation.
\begin{definition}[Refinement]
    \label{def:ref}
A component $B$ is a refinement of component $A$, written as $B \sqsubseteq A$, if and only if $E_B \subseteq E_A$ and $L_B \subseteq L_A$.
\end{definition}
\begin{lemma}
    \label{lemma:po-ref}
The relation $\sqsubseteq$ is a partial order on components.
\end{lemma}
\begin{proof}
    See \ref{appendix:proof}.
\end{proof}

An alternative to refinement is \emph{containment}. The containment relation uses a point-wise inclusion relation on observations of two TESs. The containment relation on components requires that every TES in the behavior of one is point-wise contained in a TES from the behavior of the other. 
\begin{definition}[Containment]
    \label{def:containment}
    A TES $\sigma$ is contained in a TES $\tau$, written as $\sigma \leq \tau$, if and only if, for all $i \in \N$, $\pr_1(\sigma)(i) \subseteq  \pr_1(\tau)(i)$ and $\pr_2(\sigma) = \pr_2(\tau)$.
\end{definition}
We extend the containment relation to components: a component $A = (E_A, L_A)$ is contained in a component $B = (E_B, L_B)$, written $A \leq B$, if and only if $E_A \subseteq E_B$, and for every $\sigma \in L_A$, there exists a $\tau \in L_B$ such that $\sigma \leq \tau$. 

\begin{lemma}
    \label{lemma:po-cont}
The relation $\leq$ is a pre-order over arbitrary set of components.
Let $\mathcal{C}$ be a set of components such that, for all components $A \in \mathcal{C}$ and for any two TESs $\sigma:A$ and $\tau:A$, $(\sigma \leq \tau) \implies \sigma = \tau$, then $\leq$ is a partial order on $\mathcal{C}$.
\end{lemma}
\begin{proof}
    See \ref{appendix:proof}.
\end{proof}

\begin{remark}
    The restriction in Lemma~\ref{lemma:po-cont} to consider components with no internal self containments between distinct TESs is necessary for having $\leq$ as a partial order. Consider for instance a component $A$ with only two TESs in its behavior, $\sigma:A$ and $\tau:A$ where $\pr_1(\sigma) = (\{a,b\})^\omega$ and $\pr_1(\tau) = (\{a\})^\omega$ and $\pr_2(\sigma) = \pr_2(\tau)$.
    Let $B$ be a component with a singleton behavior $\delta:B$ such that $\pr_1(\delta) = (\{a,b\})^\omega$ and $\pr_2(\delta) = \pr_2(\sigma)$.
    Then, $A \leq B$, and $B \leq A$, but $A \not = B$.
\end{remark}

\subsection{Algebraic product}
\label{sec:algebra}
A complex system typically consists of multiple components that interact with each other.
In~\cite{DBLP:journals/corr/abs-2110-02214}, we define a family of binary products acting on components, each parametrized with an interaction signature that captures the type of the interaction between a pair of components. As a
result, the product of two components, under a given interaction signature, returns a new component whose behavior
is obtained from constraining each operand component behavior according to the interaction signature.
    An interaction signature consists of two elements: a composability relation and a composition function.

A \emph{composability relation} $R(E_A, E_B) \subseteq \TES{E_A} \times \TES{E_B}$ captures \emph{what} TESs are composable.
Intuitively, when no constraints are imposed on occurrences of events in $E_A$ and $E_B$, the identity composability relation allows every pair of TESs to compose. Alternatively, if some events are shared between $E_A$ and $E_B$, only observations in $\sigma$ and $\tau$ that \emph{agree} on the occurrence or exclusion of shared events may compose, e.g., as in Example~\ref{ex:sync}. 
Also, $R$ may capture some events that cannot occur at the same time (e.g., the encapsulation operator of ACP~\cite{BERGSTRA1984109}).
A composability relation defines which pairs of TESs can compose, but a product operator must also specify how to produce a new TES from a pair of composable TESs.

A \emph{composition function} $\oplus : \TES{E_A} \times \TES{E_B} \rightarrow \TES{E_A \cup E_B}$ forms the composition $\sigma \oplus \tau$ of two TESs $\sigma$ and $\tau$.
Intuitively, the \emph{zipping} and item-wise union of the two TESs $\sigma$ and $\tau$ is one possible such composition function (such as $\cup$ in Example~\ref{ex:sync}). As well, one may construct a new event as the result of the simultaneous occurrence of some events (e.g., the communication function of ACP~\cite{DBLP:books/daglib/0000497}).

Let $A = (E_A, L_A)$ and $B=(E_B, L_B)$ be two components. We use $\Sigma = (R,\oplus)$ to range over interaction signatures, with $R$ a composability relations and $\oplus$ a composition functions.
 \begin{definition}[Product]
     \label{def:prod}
     The \emph{product} of $A$ and $B$ under $\Sigma = (R,\oplus)$ is the component $C = A \times_\Sigma B = (E_A \cup E_B, L)$  where
     \[
         L = \{ \sigma \oplus \tau \mid \sigma \in L_A, \tau \in L_B,\ (\sigma, \tau) \in R(E_A, E_B)\}
     \]
 \end{definition}

For simplicity, we write $\times$ as a general product with arbitrary interaction signature $\Sigma$.
\begin{remark}
    In general, composability relations on TESs are derived from composability conditions on observations. 
In~\cite{DBLP:journals/corr/abs-2110-02214}, we use a co-algebraic procedure to lift constraints on observations to constraints on TESs.
Such a scheme provides a wide range of possibilities for user-defined products, and a co-inductive proof principal for interaction signature equivalence.
\end{remark}

A useful interaction signature $\Sigma$ is the one that synchronizes shared events between two components. 
We write $\Sigma_{\it{sync}} = (R_{\it{sync}},\cup)$ for such interaction signature, and give its specification in Example~\ref{ex:sync}.

\newcommand{\ssync}{\mathit{sync}}
\begin{example}[Synchronous signature]
    \label{ex:sync}
\begin{table}
\centering
\caption{Three prefixes of timed-event sequences for $R(1,T,P)$, $R(2,T,P)$, and $R(3,T,P)$, together with the prefix resulting from forming their synchronous product with the swap protocol. Initially, $\mu(1) = (0;2)$, $\mu(2) = (0;1)$, $\mu(3) = (0;0)$.}
\begin{tabular}{cc}
    $t/T$ &    $\eta: R(1,P,T) \bowtie R(2,P,T) \bowtie R(3,P,T) \bowtie \swap(2,3)$\\
    \hline
    $1$   & $\{N(1,p)\}$\\
    $2$   & $\{W(1,p)\}$ \\
    $3$   & $\{W(1,p), N(2,p)\}$\\
    $4$   & $\{S(1,p), W(2,p), E(3,p)\}$\\
    $5$   & $\{S(2,p)\}$\\
    $...$ & $...$ 
\end{tabular}
\label{table:comp}
\end{table}
    In this example, we define the synchronous interaction signature $\Sigma_{\ssync} = (R_{\ssync},\cup)$.
    In a cyber-physical system, the action (of a cyber system) and the reaction (of a physical system) co-exist simultaneously in the same observation, and are therefore synchronous.

    First, we define $\cup$ that, given two TESs, returns the interleaving or the union of their observations, i.e., $(\sigma \cup \tau)(t) = \sigma(t) \cup \tau(t)$ and $\it{dom}(\sigma\cup\tau) = \it{dom}(\sigma)\cup\it{dom}(\tau)$.

Then, $R_\ssync(E_1,E_2)$ relates pairs of TESs such that all shared events occur at the same time in both TESs, i.e.,  $(\sigma,\tau) \in R_\ssync(E_1,E_2)$ if and only if, for all time stamps $t \in \Rp$, 
$\sigma(t)\cap E_2 = \tau(t) \cap E_1$.

The synchronous interaction signature $\Sigma_\ssync = (R_{\it sync}, \cup)$ leads to the product $\times_{\Sigma_\ssync}$ that forces two components to observe shared events at the same time. We write $\bowtie$ for such product.
     As a result, $R(1,P,T) \bowtie R(2,P,T) \bowtie R(3,P,T) \bowtie \swap(2,3)$ captures all sequences of moves for the robots constrained by the $\swap$ protocol.
\hfill $\blacksquare$
\end{example}
\begin{example}[Field-Robot signature]
    \label{ex:fr-sign}
    The interactions occurring between the field and the robot components impose simultaneity on some disjoint events.
    For instance,
    every observation of the robot containing the event $d(i,p) \in E_R(i,P)$ must occur at the same time as an observation of the field containing the event $\it{move(i,\overrightarrow{v(d,p)})} \in E_F(I)$ with $\overrightarrow{v(d,p)}$ returning the velocity as a function of direction $d$ and required power $p$.
    Also, every observation containing the event $\it{read(i,(\lfloor x\rfloor,\lfloor y\rfloor))} \in E_R(i,P)$ must occur at the same time as an event $(x;y)_i \in E_F(I)$ where $\lfloor z\rfloor$ gives the integer part of $z$.

    Formally, we capture such interaction in the interaction signature $\Sigma_{RF} = (R_{RF},\cup)$, where $R_{RF}$ is the smallest symmetric relation defined as for all $(\tau,\sigma)\in R_{RF}$, for all $t\in \Rp$, 
    \[\it{read(i,(n,m))} \in \tau(t) \iff  (\exists \it{(x;y)_i} \in \sigma(t) \land n = \lfloor x\rfloor \land m = \lfloor y\rfloor)\]
    and $\it{d(i,p)} \in \tau(t) \iff  \it{move(i,\overrightarrow{v(d,p)})} \in \sigma(t)$.

    As a result, the product $(R(1,T,P) \bowtie R(2,T,P) \bowtie R(3,T,P)) \times_{\Sigma_{RF}} F_\mu(I)$ captures all sequences of observations for the three robots constrained by the field component. 
    \hfill$\blacksquare$
\end{example}
\begin{remark}
    The floor part $\lfloor\cdot \rfloor$ acts as an approximation of the robot sensor on the field's position value.
    A different interaction signature may, for instance, introduce some errors in the reading.

    The interaction signature may also impose that $d(i,p)$ relates to the speed $(0,1/T)$, $(0,-1/T)$, $(-1/T,0)$, and $(1/T,0)$ when $d = N$, $d = S$, $d = W$, and $d = E$, respectively. Then, for a time interval $T$, the power $p$ moves the robot by one unit on the field.
\end{remark}
\begin{remark}
  In practice, it is unlikely that two observations happen at exactly at the same time.
  However, in our framework, the time of an observation is an abstraction that requires every event of the observation to
occur after the events of the previous observation, and before the events of the next observation.
\end{remark}

\begin{definition}[Monotonicity]
    \label{def:monotonicity}
    Let $\times$ be a commutative product.
    Then, $\times$ is monotonic if and only if, for $B \sqsubseteq A$ and for any $C$, we have $B \times C \sqsubseteq A \times C$.
\end{definition}
\begin{remark}
Monotonicity shows that the inclusion of component's behavior is preserved by product. 
Let $A$ and $B$ be two components such that $B \sqsubseteq A$.
Suppose that $P$ is a component that models a property satisfied by component $A$ and preserved under product with a component $C$, then $P$ is satisfied by component $B$ and component $B\times C$, by monotonicity. 
Note that the definition of monotonicity assumes $\times$ to be commutative. That assumption can be relaxed by defining left and right monotonicity. 
\end{remark}

\begin{lemma}[Monotonicity of $\bowtie$]
    The product $\bowtie$ in Example~\ref{ex:sync} is monotonic.
    \label{lemma:comp}
\end{lemma}
\begin{proof}
    See \ref{appendix:proof}.
\end{proof}

\section{Division and Conformance}
\label{sec:division}
Consider two components $B$ and $C$, and a product $\times$ over components that models the interaction constraints between $B$ and $C$. We use $=$ on components as strict structural equality: $A = B$ if the interfaces and behaviors of $A$ and $B$ are equal.
The composite expression $A = C \times B$ captures, as a component, the concurrent observations of components $C$ and $B$ under the interaction modelled by $\times$. 
Consider a component $D$ different from $C$ such that $C \times B = D \times B$. The equality states that the result of $D$ interacting with $B$ is the same as $C$ interacting with $B$. 
Consequently, in this context, component $C$ can be replaced by component $D$ while preserving the global behavior of $A$.

In general, a component $D$ that can substitute for $C$ is not unique. The set of alternatives for $C$ depends, moreover, on the product $\times$, on the component $B$, and on the behavior of $A$.  A `goodness' measure may induce an order on this set of alternative components, and eventually give rise to a \emph{best} substitution candidate.
More generally, the problem is to characterize, given two components $A$ and $B$ and an interaction product $\times$, the set of all components $C$ such that $A = C \times B$.

\subsection{Divisibility and quotient} 
The divisibility of a component $A$ by a component $B$ under a product $\times$ captures the possibility to write $A$ as a product of $B$ with another component. 

\begin{definition}[Right (left) divisibility]
    \label{def:div}
    A component $A$ is right (respectively, left) divisible by $B$ under a product $\times$ if there exists a component $C$ such that  $B \times C = A$ (respectively, $C \times B = A$).
\end{definition}

$A$ is \emph{divisible} by $B$ under $\times$ when $A$ is both left and right divisible by $B$ under $\times$.
Intuitively, the set of witnesses for divisibility, contains every component whose product (under the same interaction signature) with the divisor yields the dividend. We call such witnesses \emph{quotients}.

\begin{definition}[Right (left) quotients]
    \ The right (respectively, left) quotients of $A$ by $B$ under the product $\times_\Sigma$, written $A/^*_\Sigma B$ (respectively, $A\setminus^*_\Sigma B$), is the set $\{C \mid B \times_\Sigma C = A \}$ (respectively, $\{C \mid C \times_\Sigma B = A \}$).
\end{definition}

If $\times_\Sigma$ is commutative, then $A/^*_\Sigma B = A\setminus^*_\Sigma B$, in which case we write $_\Sigma\cfrac{A}{B}*$.
We define left (right) division operators that pick, given a choice function\footnote{We assume the axiom of choice~\cite{Plato} and the existence of a function $\chi$ that picks an element from a set.}, the best element from their respective sets of quotients as their quotients.

\begin{definition}[Right (left) division]
    Let $A$ be divisible by $B$ under $\times_\Sigma$.
    The right (respectively, left) quotient of $A$ divided by $B$, under the product $\times_\Sigma$ and the choice function $\chi$ over the right (respectively, left) quotients, is the element $\chi(A/_\Sigma^* B)$ (respectively, $\chi(A\setminus_\Sigma^* B)$). We write $A/_\Sigma^\chi B$ (respectively, $A\setminus_\Sigma^\chi B$) to represent the quotient.
\end{definition}

If $\times_\Sigma$ is commutative, then $A/^\chi_\Sigma B = A\setminus^\chi_\Sigma B$, in which case we denote the division as $_\Sigma\cfrac{A}{B}\chi$.

\begin{example}[Lowest element]
    One measure 
    that can impose an order on a set of quotients uses the fact that a component may contain all behavior of another.
    Indeed, every quotient
    has the property that, in composition with the divisor, yields the dividend. Then, a quotient whose behavior is fully contained in that of every other quotient may be optimal in terms of behavior complexity.

    Let $\mathcal{C}$, the set of right (left) quotients for $A$ divisible by $B$ for product $\times$, be equipped with an ordering such that the lowest element is an element of $\mathcal{C}$, then a function that picks the lowest element can act as a choice function to define the result of the division of $A$ by $B$.
    \hfill$\blacksquare$
\end{example}

One may consider $\leq$ as a natural ordering on quotients.
However, the set of quotients equipped with the containment relation may not have a lowest element.
One such example is shown in Table~\ref{table:counter-ex-quo}. 
Consider $A$, $B$, $C$, and $D$ with $\{0,1,2\}$, $\{0,1\}$, $\{0,2\}$, and $\{1,2\}$ as interface, respectively.
Using the synchronous composition operation, the TESs $\tau$ and $\eta$ compose with the TES $\delta$ to give the TES $\sigma$. However, $C$ and $D$ require synchronization on their shared event to compose with $B$. A smaller component than $C$ and $D$ would be a component $F$, whose interface is the singleton set containing event $2$. 
However, such component has no shared event with $B$, and may therefore freely interleave its observations, which does not correspond with observations in $A$. Thus, $F$ is not an element of the quotients, and $C$ and $D$ have no lower bound in the set of quotients.

\begin{table}
\centering
\caption{Counter example for a lowest element in the division of $A$ by $B$, with $C$ and $D$ two quotients.}
\begin{tabular}{ccccc}
          &\ $\sigma:A$\ &\ $\tau:B$\ &\ $\delta:C$ \ &\ $\eta: D$\\
    \hline
    $t_1$ & $\{0,1,2\}$ & $\{0,1\}$ & $\{0,2\}$  & $\{1,2\}$ \\
    $t_2$ & $\{0,1,2\}$ & $\{0,1\}$& $\{0,2\}$   & $\{1,2\}$ \\
    $t_3$ & $\{0,1,2\}$  & $\{0,1\}$ & $\{0,2\}$ & $\{1,2\}$ \\
    $...$ & $...$ & $...$          & $...$ &
\end{tabular}
\label{table:counter-ex-quo}
\end{table}

We show in the next theorem that a subset of quotients with a shared interface has a lower bound. 
We discuss how the choice of an interface for a quotient may be guided by some qualitative design choices. 

\newcommand{\lub}{\mathit{lub}}
\newcommand{\glb}{\mathit{glb}}
\begin{theorem}
    \label{th:div}
    Let $\leq$ be the containment relation introduced in Definition~\ref{def:containment}.
    Let $\times_\Sigma$ be a commutative, associative, and idempotent product on components, and such that for any two components $C$ and $D$ with the same interface, $C \times_\Sigma D \leq C$.
    Given $A$ divisible by $B$ under $\times_\Sigma$, any finite subset of quotients sharing the same interface $E$ has a lower bound that is itself a quotient in $A/^*_\Sigma B$.
\end{theorem}
\begin{proof}
Let $\mathcal{C}(E)$ be a finite subset of the set $\{ C \mid C$ \textit{ has interface $E$ and } $C \in A/^*_\Sigma B \}$.
    We also write $\times_\Sigma \mathcal{C}(E)$ for the product of all components in $\mathcal{C}(E)$.

    For any $C \in \mathcal{C}(E)$, we have
    \[ \times_\Sigma \mathcal{C}(E) \leq C\]
    which makes $\times_\Sigma \mathcal{C}(E)$ a lower bound for $\mathcal{C}(E)$.

    Given associativity, commutativity, and idempotency of $\times_\Sigma$, for any $C_1, C_2 \in \mathcal{C}(E)$:
    \begin{align*}
        A & = B \times_\Sigma C_1 \\
        A & = B \times_\Sigma C_2 \\
        A\times_\Sigma A  = A & =  (B \times_\Sigma C_1) \times_\Sigma (B \times_\Sigma C_2) \\
                                  A & = B \times_\Sigma (C_1 \times_\Sigma C_2)
    \end{align*}
    which, applied over the set $\mathcal{C}(E)$, gives $A = B \times_\Sigma (\times_\Sigma \mathcal{C}(E))$.
    Thus, $\times_\Sigma \mathcal{C}(E) \in \mathcal{C}(E)$. 
\footnote{Strictly speaking, closure under finite product does not necessarily imply closure under infinite product. We leave investigating the conditions under which closure under infinite product holds, for future work.}
    \qed
\end{proof}

\newcommand{\one}{\mathbf{1}}
\newcommand{\zero}{\mathbf{0}}
When conditions of Theorem~\ref{th:div} are satisfied, we write $A/^{\leq,E}_\Sigma B$ for the lower bound of the set of quotients with interface $E$.
\begin{remark}
    The operation of division defined by Theorem~\ref{th:div} raises several points for discussion.  
    First, the set of quotients sharing the same interface is structured. Indeed, when the interface is fixed, each finite subset of quotients has a lowest element under $\leq$, which makes the definition of a division operator possible. 
    Second, the fact that there is, in general, no minimal element over the set of all quotients reveals the important role that interfaces play in system decomposition. 
    In other words, one may consider another measure to \emph{choose} a quotient interface, that is orthogonal to behavior containment (see Section~\ref{sec:discussion} for a discussion about the cost of coordination).
\end{remark}
 
We use $\one$ to denote the component $(\emptyset,\TES{\emptyset})$, and $\zero$ to denote the component $(\emptyset, \emptyset)$, that has the empty interface and no behavior.

A component $A = (E_A, L_A)$ is \emph{closed under insertion of silent observations} if, for any $\sigma \in L_A$, and for any silent observation $(\emptyset, t)$ with $t \in \Rp$, and given $i \in \N$ such that $\sigma(i) = (O,t_1)$ and $\sigma(i+1) = (O',t_2)$ with $t_1 < t < t_2$, then there exists $\tau \in L_A$ such that $\sigma(k) = \tau(k)$ for all $k \leq i$, $\sigma(i+1) = (\emptyset, t)$, and $\sigma(k+2) = \tau(k+1)$ for all $k>i$.

In order to reason about components algebraically, we want some properties to hold. For instance, that a component is divisible by itself and the set of quotients contains the unit element.
\begin{lemma}
    Let $A$ be a component closed under insertion of silent observations, and $\Sigma_\ssync$ the synchronous interaction signature introduced in Example~\ref{ex:sync}. Then, $\one \in A/^*_{\Sigma_\ssync} A$.
    \label{lemma:unit}
\end{lemma}
\begin{proof}
    For any element $\sigma:A$, and for any $\tau:\one$, we have $(\sigma,\tau) \in R$ and $\sigma [\cup] \tau: A$. Moreover, for any $\sigma:A$, there exists $\tau:\one$ such that $(\sigma,\tau) \in R$ and $\sigma [\cup] \tau = \sigma$. Then, $\one$ is in the set of quotients of $A$ by $A$.
    \qed
\end{proof}

\begin{remark}
    Note that Lemma~\ref{lemma:unit} assumes components to be closed under insertion of silent observations. The reason, as shown in the proof, comes from the product of $\one$ with a component $A$ that may insert silent observations at arbitrary points in time. 
    A consequence of Lemma~\ref{lemma:unit} is the existence of a choice function that can pick, from the set of quotients, the unit component for the division of $A$ by $A$.
\end{remark}
\begin{example}
    Let $(R(1,P,T) \bowtie R(2,P,T) \bowtie R(3,P,T)) \times_{\Sigma_{RF}} F_\mu(I)$ be the product of three robot components and a field component with $I = \{1,2,3\}$.
    Consider the component $P = (E,L)$ with $E = \{\it{read((n,m),i)}, (n,m)_i \mid n,m \in \N \}$ and $L \subseteq \TES{E}$.

    Then, $((R(1,P,T) \bowtie R(2,P,T) \bowtie R(3,P,T)) \times_{\Sigma_{RF}} F_\mu(I))/_{\Sigma_{RF}}^{\leq, E'} P$, with $E' = (E_R(1) \cup E_R(2) \cup E_R(3) \cup E_F(I)) \setminus \{ (n,m)_i \mid n,m \in \N  \}$, denotes the component that, in composition with $P$, recovers the initial system. Note that the component resulting from division ranges over the interface $E'$. As a consequence, all events $(n,m)_i$ have been hidden in the quotient. 
    Note that the division exists due to the interaction signature $\Sigma_{RF}$ that imposes simultaneity on occurrence of events $\it{read((n,m),i)}$ and $(n,m)_i$.
    \hfill$\blacksquare$
\end{example}

\subsection{Conformance}
\label{sec:conformance}
The criterion for divisibility of $A$ by $B$, under product $\times$, is the existence of a quotient $C$ such that $B \times C= A$.
The equality between $B \times C$ and $A$ makes division a suitable decomposition operator. We can define, a similar operation to describe all components $C$ that \emph{coordinate} $B$ in order for the result to behave in conformance with specification $A$. In this case, we replace equality with the refinement relation of Definition~\ref{def:ref}.
We consider in the sequel non-empty coordinators, i.e., components different from $\bf{0} = (\emptyset, \emptyset)$.

\begin{definition}[Right (left) conformance]
    \label{def:conf}
    Component $B$ is right (respectively, left) conformant to component $A$ under $\times$ if there exists a non-empty component $C$ such that  $C \times B \sqsubseteq A$ (respectively, $B \times C \sqsubseteq A$).
\end{definition}

\begin{definition}[Right (left) conformance coordinators]
    \label{def:star-conf}
    \ \ The set of right (respectively, left) conformant coordinators that make $B$ behave in conformance with $A$ under $\times_\Sigma$, denoted as $A \downharpoonright^*_\Sigma B$ (respectively, $A \downharpoonleft^*_\Sigma B$), is the set $\{C \mid C \not = \mathbf{0} \textit{ and } C \times_\Sigma B \sqsubseteq A \}$ (respectively, $\{C \mid C \not = \mathbf{0} \textit{ and } B \times_\Sigma C \sqsubseteq A \}$).
\end{definition}

If $\times_\Sigma$ is commutative, then $A \downharpoonright^*_\Sigma B = A \downharpoonleft^*_\Sigma B$, in which case we write $A \downarrow^*_\Sigma B$. 
Similarly to the set of quotients, the set of coordinators having the same interface is structured and gives ways to define non-trivial coordinators, as in Theorem~\ref{th:conf}.

\begin{definition}[Right (left) principal coordinator]
    \label{def:coord}
    Let $B$ be confor-mable to component $A$, and
    let $\chi$ be a choice function that selects the best component out of a set of components.
    The right (respectively, left) principal coordinator that makes $B$ behave in conformance with $A$, denoted as $A \downharpoonright^\chi_\sigma B$ (respectively, $A \downharpoonleft^\chi_\Sigma B$), is the component $\chi(A \downharpoonright^*_\Sigma B)$ (respectively, $\chi(A \downharpoonleft^*_\Sigma B)$).
\end{definition}

\begin{example}[Greatest element]
    \label{ex:greatest}
    One measure 
    that can impose an order on a set of coordinators uses containment.
    The refinement relation used to define conformance also accepts coordinators that have no behavior at all, and trivially satisfies the behavior inclusion relation.
    Maximizing the observables of the resulting composite behavior set, corresponds to finding the greatest coordinator under a containment relation. 

    More generally, if $\mathcal{C}$, the set of right (left) coordinators that make $B$ conformant to $A$ under $\times$, is equipped with an ordering such that the greatest element is an element of $\mathcal{C}$, then the function that picks the greatest element can act as a choice function to select the best conformance coordinator of $B$ to behave as $A$ under $\times$.
    \hfill$\blacksquare$
\end{example}

Following the result of Theorem~\ref{th:div}, if the interface of the quotient is fixed, then the subset of quotients that have the same interface has a least element with the containment relation introduced in Definition~\ref{def:containment}. 
We show in Theorem~\ref{th:conf} that a similar result holds for the set of coordinators.
\begin{theorem}
    \label{th:conf}
    Let $\leq$ be the containment relation introduced in Definition~\ref{def:containment}.
    Let $\times_{(R,\oplus)}$ be a commutative, associative, idempotent, and monotonic (as in Definition~\ref{def:monotonicity}) product on components.
    Given $B$ conformant with $A$ under $\times_{(R,\oplus)}$, any finite subset of coordinators sharing the same interface $E$ has an upper bound that is itself a coordinator in $A\downarrow^*_\Sigma B$.
\end{theorem}

Finding a conformance coordinator that makes $B$ behave in conformance with $A$ is looser than finding a quotient for $A$ divisible by $B$: any quotient of $A$ by $B$ under a product $\times_\Sigma$ is therefore a coordinator that makes $B$ conformant with $A$.
Such quotient-coordinator has the property that it ``coordinates'' B such that the resulting behavior covers the entire behavior of $A$.

For some suitable products, Theorem~\ref{th:div} and Theorem~\ref{th:conf} state the existence of, respectively, a lowest element in the subsets of quotients and a largest element in the set of coordinators that share the same interface. 
The synchronous product introduced in Example~\ref{ex:sync} is one product that satisfies the requirements of both theorems.
We show in Section~\ref{sec:sync} how division under synchronous product can characterize valid \emph{updates} of a system, and how conformance under synchronous product can characterize valid \emph{protocols}.

\newpage
\section{Applications of Division}
\label{sec:sync}

In this section, we consider the robot, field, and protocol components introduced in Examples~\ref{ex:robots},~\ref{ex:grid}, and~\ref{ex:protocol}, together with the synchronous product $\bowtie$ of Example~\ref{ex:sync} and the product $\times_{\Sigma_{RF}}$ of Example~\ref{ex:fr-sign}.
Both products are commutative (Lemma $1$ in~\cite{DBLP:journals/corr/abs-2110-02214}), and we therefore omit the right and left qualifiers for division and conformance.

\paragraph{Initial conditions}
For each robot, we fix the power requirement of a move and the time period $T$ between two observations to be such that a move of a robot during a period $T$ corresponds to a one unit displacement on the field. Then, each move action of a robot changes the location of the robot by a fixed number of units or none if there is an obstacle. We write $R(i)$ for robot $R(i,P,T)$ with such fixed $P$ and $T$.
As an example, the observation $(\{d(i),\mathit{read}(i,(x;y))\}, t)$ followed by the observation $(\{\it{read}(i,(x';y'))\},t+T)$ gives only few possibilities for $(x';y')$: either $(x;y) = (x';y')$, in which case the robot got blocked in the middle of its move, or $(x';y')$ increases (or decreases) by one unit the $x$ or $y$ coordinates, according to the direction $d$.

Let the initial state $\mu$ of the field be such that $\mu(1) = (3;0)$, $\mu(2)= (2;0)$, and $\mu(3)= (1;0)$, which defines the initial positions of $R(1)$, $R(2)$ and $R(3)$ respectively, and let there be obstacles throughout the field on the $3\times 2$ rectangle from $(0;-1)$ to $(4;2)$, i.e., for all $(x,y) \in (\llbracket 0;4\rrbracket \times \llbracket -1;2\rrbracket) \setminus (\llbracket 1;3\rrbracket \times \llbracket 0,1\rrbracket)$, there exists $i \in I$ such that $\mu(i) = (x,y)_i$. As a result, the moves of each robot are restricted to the inside of the $3\times 2$ rectangle as displayed in Table~\ref{table:strategy}.

\subsection{Approximation of the Field as a Grid}
\paragraph{Problem}
A field component captures in its behavior the continuous responses of a physical field interacting with robots roaming on its surface. The interface of the field contains therefore an event, per object, for each possible position and each possible move. 
In some cases, however, only a subset of those events are of interest. For instance, we may want to consider only integer position of objects on the grid, and discard intermediate observables. As a result, such component would describe a discrete grid instead of a continuous field, while preserving the internal physics: no two objects are located on the same position.
We show how to define the grid as a subcomponent of the field, using the division operator.

\paragraph{Definition of the grid}
We use division to capture a discrete grid component $G_\mu(I)\leq F_\mu(I)$ contained in the field component $F_\mu(I)$. A grid component has the interface $E_G(I)$, where $E_G(I) \subseteq E_F(I)$ with $(x,y)_i \in E_G(I)$ implies $x,y \in \N$.

We use the component $C = (E_G(I), \TES{E_G(I)})$ to denote the free component whose behavior contains all TESs ranging over the interface $E_G(I)$.
Then, by application of Theorem~\ref{th:div}, we use the least element with respect to $\leq$ of the set of quotients of $C \times_{\Sigma_{\ssync}} F_\mu(I)$ divided by $F_\mu(I)$ under $\Sigma_\ssync$ to define the grid. Thus,
\begin{equation}
    \label{eq:grid}
    G_\mu(I) = _{\Sigma_\ssync}\cfrac{C \times_{\Sigma_{\ssync}} F_\mu(I)}
    {F_\mu(I)}(\leq, E_G(I))
\end{equation}
which naturally emerges as a subcomponent of the field component $F_\mu(I)$.

\paragraph{Consequences}
The grid component inherits some physical constraints from the field $F_\mu(I)$, but is strictly contained in the field component.
There is a fundamental difference between an approximation of the position as a robot sensor detects, and a restriction of the field to integer positions as in the grid component. In the former, the component reads a value that does not corresponds precisely to its current position, while in the latter, the position read is exact but observable only for integer values.

As a result, the two component expressions $(R(1)\bowtie R(2) \bowtie R(3)) \times_{\Sigma_{RF}} G_\mu(I)$ and $(R(1)\bowtie R(2) \bowtie R(3)) \times_{\Sigma_{RF}} F_\mu(I)$ restrict each robot behavior in different ways: the grid component allows discrete moves only and the position that a robot reads is the same position as that of an object on the grid, while the field component allows continuous moves but the position that a robot reads is an approximation of the position of the robot on the field.
In the sequel, we use the grid component $G_\mu(I)$ instead of the field component.

\subsection{Updates of components}
\label{sec:updates}
\paragraph{Problem} 
The interaction signature of a product operator on components restricts which pairs of behaviors are composable.
As a consequence, some components may have \emph{more behavior} than necessary, namely the elements that do not occur in any composable pair. An update is an operation that preserves the global behavior of a composite system while changing an operand of a product in the algebraic expression that models the composed system. The goal of such update, for instance, is to remove some behaviors that are not composable or prevent some possible runtime errors. 
We give an example of such update that replaces a robot component by a new version that removes some of its possibly blocking moves.

\paragraph{Scenario}
For each robot, we fix its behavior to consist of TESs that alternate between move and reading observations. Moreover, for a robot's period $T$, and arbitrary $n_i \in \N$, we let $T \times n_i, i\in\N$, represent the timestamp of the $i^{th}$ observation of a TES in its behavior, so long as $n_i<n_{i+1}$.
Table~\ref{table:counter-ex} displays elements of the behavior for each robot. For instance, the TES $\eta:R(1)$ captures the observations resulting from $R(1)$ moving west twice. Note that, in composition with the grid component, the readings may conflict with the actual position of the robot, as some moves may not be allowed due to obstacles on the path.

For instance, given the expression $(R(1)\bowtie R(2)\bowtie R(3))\times_{\Sigma_{RF}} G_\mu(I)$, the TES $\eta$ is not observable as it is not composable with any of the TESs $\tau$ or $\delta$ from $R(2)$ and $R(3)$ respectively. We show how to use division to remove all of such behaviors.

\begin{table}
\centering
\caption{Prefixes of four TESs for $R(1)$, $R(2)$, and $R(3)$. For direction $d$ and robot $i$, we write $d(i)$ instead of $d(i,p)$ since the power $p$ is initially fixed. We omit the set notation as observations are all singletons. We consider $(n_i)_{i\in\N}$ as an increase sequence of natural numbers.} 
\begin{tabular}{ccccc}
    t/T &\ $\sigma:R(1)$\ &\ $\eta:R(1)$\ &\ $\tau:R(2)$\ &\ $\delta:R(3)$ \\
    \hline
    $n_0$ & $\rread(1,(3;0))$  & $\rread(1,(3;0))$  & $\rread(2,(2;0))$ & $\rread(3,(1;0))$  \\
    $n_1$ & $N(1)$  & $W(1)$  & $N(2)$ & $E(3)$  \\
    $n_2$ & $\rread(1,(3;1))$  & $\rread(1,(2;0))$  & $\rread(2,(2;1))$ & $\rread(3,(2;0))$  \\
    $n_3$ & $W(1)$  & $W(1)$  & $W(2)$ & $E(3)$ \\
    $n_4$ & $\rread(1,(2;1))$  & $\rread(1,(1;0))$  & $\rread(2,(1;1))$ & $\rread(3,(3;0))$  \\
    $n_5$ & $W(1)$  & $\emptyset$  & $S(2)$ & $\emptyset$ \\
    $n_6$ & $\rread(1,(1;1))$  & $\emptyset$  & $\rread(2,(1;0))$ & $\emptyset$  \\
    $n_7$ & $S(1)$  & $\emptyset$  & $E(2)$ & $\emptyset$ \\
    $n_8$ & $\rread(1,(1;0))$  & $\emptyset$  & $\rread(2,(2;0))$ & $\emptyset$  \\
    $n_9$ & $\emptyset$ & $\emptyset$ & $\emptyset$ & $\emptyset$ \\
    $...$ & $...$       & $...$       & $...$          & $...$
\end{tabular}
\label{table:counter-ex}
\end{table}

\paragraph{Update} 
The replacement for $R(1)$ should preserve the global behavior.
We use division to define an update $R'(1)$ of $R(1)$ that removes all elements from its behavior that are not composable with any element from the behavior of $R(2)$ and $R(3)$ under the constraints imposed by the grid.

As a result, the component
\begin{equation}
    \label{eq:star-div}
    R'(1) = \Sigma_{\ssync} \cfrac{(R(1)\bowtie R(2)\bowtie R(3)) \times_{\Sigma_{RF}} G_\mu(I)}{(R(1)\bowtie R(2)\bowtie R(3))\times_{\Sigma_{RF}} G_\mu(I)}(\leq, E_R(1))
\end{equation}
contains in its behavior all elements ranging over the interface $E_R(1)$ that are composable with elements in the behavior of the dividend component.
Note that the set of quotients is filtered on the interface $E_R(1)$, and $R(1)$ trivially qualifies as a quotient. However, $R(1)$ is not minimal as $\eta$ can be removed from its behavior.

\paragraph{Consequence} As a consequence, we defined, using division, an update for component $R(1)$ that removes some elements of its behavior while preserving the global behavior of the composite expression.

Note that the fact that in our example each robot alternates between a move and a read is crucial to remove, by composition, undesired behavior. Indeed, the readings of each robot must synchronize with the location displayed on the grid, and therefore implies that the robot successfully moved. The constraints imposed by the grid coordinate the robot by preventing two robots to share the same location.

\newcommand{\sorted}{\mathit{sorted}}
\subsection{Coordination and distribution}
\label{sec:protocol}
\paragraph{Problem}
Consider the scenario previously described with an additional modification: a robot no longer observes its location after every move, but only at the end of the sequence of moves. The TESs of each robot's behavior are described in Table~\ref{table:strategy}, where $(x_i,y_i)$ ranges over possible position readings for robot $i$. As a result, conflicts between robots may no longer be observable, and the timing of observations may render some incidents of robots blocking each other unobservable. 
We define a coordinator that makes the system conformant to a global property. As opposed to the division operation, a conformance coordinator may restrict the system behavior to a subset that conforms to a specified property.
We consider the following property $P_\sorted(I)$: ``eventually, all the robots get sorted, i.e., every robot $R(i)$ eventually ends on the grid location $(i;0)$.''.

\begin{table}
\centering
\caption{Prefixes of three TESs for $R(1)$, $R(2)$, and $R(3)$, graphically represented by some trajectories on a grid.}
\begin{tabular}{cccc}
    t/T &\ $\sigma:R(1)$\ &\ $\tau:R(2)$\ &\ $\delta:R(3)$ \\
    \hline
    $n_1$ & $N(1)$  & $N(2)$ & $E(3)$  \\
    $n_2$ & $W(1)$  & $W(2)$ & $E(3)$ \\
    $n_3$ & $W(1)$  & $S(2)$ & $\rread(3,(x_3;y_3))$ \\
    $n_4$ & $S(1)$  & $E(2)$ & $\emptyset$ \\
    $n_4$ & $\rread(1,(x_1;y_1))$  & $\rread(2,(x_2;y_2))$ & $\emptyset$ \\
    $n_5$ & $\emptyset$ & $\emptyset$ & $\emptyset$ \\
    $...$ & $...$       & $...$          & $...$
\end{tabular}
\qquad\qquad
    \begin{tikzpicture}
        [baseline=(current bounding box.center)]
        \node (a) at (0,0) {};
        \node (b) at (3,2) {};
        \node[shift={(-0.5,0.5)}] (R1) at (3,0) {$R_1$};
        \node[shift={(-0.5,0.5)}] (R2) at (2,0) {$R_2$};
        \node[shift={(-0.5,0.5)}] (R3) at (1,0) {$R_3$};
        \draw[step=1.0,black,thin,xshift=1cm,yshift=1cm] (a) grid (b);
        \draw[->,draw=blue] (1.5,0.8) -- ++(0,0.5) -- ++(-0.9,0) -- ++(0,-0.5) -- ++(0.6,0) ;
        \draw[->,draw=red]  (0.4,0.3) -- ++(0,-0.2) -- ++(2.1,0) -- ++(0,0.2);
        \draw[->,draw=green] (2.5,0.8) -- ++(0,0.8) -- ++(-2.2,0) -- ++(0,-0.8);
    \end{tikzpicture}
\label{table:strategy}
\end{table}

\paragraph{Global coordinator}
We can define, from the sort property, a component as $C_\sorted(I) = (E_\sorted(I), L_\sorted(I))$ whose interface is the union of the interfaces of all robots and the grid, i.e., $E_\sorted(I) = E_G(I) \cup \bigcup_{i\in I} E_R(i)$, and whose behavior $L_\sorted(I) \subseteq \TES{E_\sorted(I)}$ contains all sequences of moves that make the robots eventually end in their respective sorted grid positions, i.e., $\sigma \in L_\sorted(I)$ if and only if there exists $t \in \Rp$ such that $(O,t) \in \sigma$ with $(i;0)_i \in O$ for all $i \in I$.
Note that, by construction, the behavior of component $C_\sorted(I)$ may contain some TESs from the behavior of component $(R(1)\bowtie R(2)\bowtie R(3)) \times_{\Sigma_{RF}} G_\mu(I)$, namely ever TESs that satisfies the property. 

Consequently, the product of component $C_\sorted(I)$ with $(R(1)\bowtie R(2)\bowtie R(3))\times_{\Sigma_{RF}} G_\mu(I)$, under the signature $\Sigma_\ssync$, defines a component whose behavior contains all elements in the behavior of $(R(1)\bowtie R(2)\bowtie R(3))\times_{\Sigma_{RF}} G_\mu(I)$ that are also in the behavior of $C_\sorted(I)$.
Therefore, if the behavior of component $((R(1)\bowtie R(2)\bowtie R(3))\times_{\Sigma_{RF}} G_\mu(I)) \bowtie C_\sorted(I)$ is not empty, $(R(1)\bowtie R(2)\bowtie R(3))\times_{\Sigma_{RF}} G_\mu(I)$ is conformant to $C_\sorted(I)$ and $C_\sorted(I)$ is a principal coordinator.
However, using $C_\sorted(I)$ as a coordinator requires each component to synchronize, at each step, with every other component. 
We show how to define a different choice function on the set of coordinators, in order to identify a minimalist form of coordination.

\paragraph{Minimalist coordinator} 
We define a coordinator whose interface is strictly contained in the interface of the global $C_\sorted(I)$ coordinator. More precisely, we search for a coordinator over the interface of robot $R(1)$ that makes the system $(R(1)\bowtie R(2) \bowtie R(3)) \times_{\Sigma_{RF}} G_\mu(I)$ conformant to the property component $C_\sorted(I)$. 
First, observe that the set of coordinators 
\[
    (R(1)\bowtie R(2)\bowtie R(3)) \times_{\Sigma_{RF}} G_\mu(I) \downarrow^*_{\Sigma_{\ssync}} C_\sorted(I)
\]
filtered on the interface $E_R(1)$ is empty. Indeed, for any set of timestamp factors $n_1$, $n_2$, $n_3$, and $n_4$ for the observables of $R(1)$ in Table~\ref{table:strategy}, there exists an element from the behavior of $R(2)$ that delays its first action until after $n_3$, and eventually ends up in a blocking position. 
As a consequence, there is no coordinator restricted to the events of $R(1)$ that makes $(R(1)\bowtie R(2)\bowtie R(3)) \times_{\Sigma_{RF}} G_\mu(I)$ conformant to $C_\sorted(I)$.

Instead, we consider filtering the set of coordinators with the interface $E_R(1)\cup \{N(2)\}$.
In this case, one can find a simple coordinator that makes the set of robots to conform with the sort property. Indeed, every observation $N(1)$ of robot $R(1)$ must occur after an observable $N(2)$ of robot $R(2)$. As a result, such coordinator $C_{12}$ restricts $R(1)$ to move only after $R(2)$ moves, which results in a composite system conformant to $C_\sorted(I)$. Thus, the new coordinator $C_{12}$, in product with $R(1)$, $R(2)$, $R(3)$, and the grid $G_\mu(I)$ satisfies the sort property, i.e., $((R(1)\bowtie R(2) \bowtie R(3))\times_{\Sigma_{RF}} G_\mu(I)) \bowtie C_{12} \sqsubseteq C_\sorted(I)$.

\paragraph{Consequence}
We showed one global coordinator for the set of robots to satisfy the sort property, and one minimalist coordinator over the interfaces of $E_R(1)$ with an event from $E_R(2)$. The minimalist coordinator has an interface strictly included in the global coordinator, which therefore minimizes the amount of interaction among components.
In the next section, we discuss the cost of coordination as a possible measure to order a set of conformance coordinators.

\section{Discussion}
\label{sec:discussion}
The operations of division and conformance defined earlier characterize all possible updates and coordinators for a composite system. In general, the set of quotients or coordinators is not a singleton, which then necessitates a choice function to \emph{pick} a component that best suits the needs. 
We saw in Theorem~\ref{th:div} how such a choice function can be defined using an ordering on components, and choosing the least among such components.
Intuitively, such a choice function prefers a quotient with the least number of observations. However, Theorem~\ref{th:div} assumes a fixed interface, and does not discuss how to rank components according to their interfaces. We discuss alternative rankings, below, in the case of the synchronous product $\bowtie$ of Example~\ref{ex:sync}.

\subsection{Cost of coordination}
Let $A = (E_A, L_A)$ and $B = (E_B, L_B)$ be two components such that the set of quotients of $A$ divisible by $B$ under $\bowtie$, the synchronous product, is non empty.
Consider, as well, a component $C = (E_C, L_C)$ in the same set of quotients.
We discuss alternative scenarios based on the interface of component $C$.

Consider the case where $E_C \cap E_B \not = \emptyset$. Then, events in the intersection $E_C \cap E_B$ are events for which $C$ and $B$ must perform a simultaneous observation, i.e., with equal time stamps. The size of the intersection $E_C \cap E_B$ can therefore characterize how much coordination must take place between two implementations of components $B$ and $C$ to successfully achieve those synchronous observations. Alternatively, if a component $D = (E_D, L_D)$ is in the set of quotients such that $|E_D \cap E_B| < |E_C \cap E_B|$, then the smaller number of shared events of $D$ with $B$ hints at a potentially smaller amount of coordination between components $D$ and $B$. 
\footnote{
     Of course, a smaller or equal number of shared interface events does not necessarily mean
smaller or equally intense coordination: it is the frequency of the occurrence of shared events
that determines the overhead of required coordination. More generally, it is the frequency of
occurrences of events (shared or otherwise) related by composability relations that impact
the coordination overhead.  The number of shared events may still
be used as a crude approximation, or some more sophisticated, e.g., stochastic, model may be
used instead, as a measure of the overhead of coordination.
}

In the case that $E_C \cap E_B = \emptyset$ and $E_C \not = \emptyset$, the family of quotients with interface $E_C$ are particularly useful. The fact that the two interfaces are disjoint means that $A$ can be decomposed into two components that do not need any coordination. Indeed, as $\bowtie$ constrains only occurrences of shared events, if $B$ and $C$ share no events then they can be run completely independently of one another. 

The two cases highlighted above give us some insight into how the intensity of coordination can be used as a measure to rank components. Note that such ranking is contextual to the dividend and the divisor. Although $C$ may require more coordination than $D$ to synchronize with $B$ in order to form component $A$, in the context of $A$ divisible by another component $F$, component $C$ may become preferable to $D$.

The cost of coordination discussed here is orthogonal to efficiency measures discussed in Theorem~\ref{th:div}. Thus, the two measures can be combined to first rank components in terms of their interfaces to minimize the amount of coordination required, and then rank components sharing the same interface in terms of the size of their behavior.

\subsection{Series of division}
We defined a division operator for components. We saw that, under some criteria, division may return a `better' description of a composite system, i.e., for $A = B \times C$, the division of $A$ by $B$ may return a $D$ better than $C$ while preserving the behavior of $A$.

The question of convergence, then, naturally follows.
Consider the expression $A = B \times C$, and the division of $A$ by $B$ returning a component $C' \not = C$.
Symmetrically, the division of $A$ by the new component $C'$ may return a component $B' \not = B$.
Repeating the same process, dividing $A$ by $B'$ and so on, produces a sequence of components $C^{(n)}$ and $B^{(n)}$ in their respective $n^{th}$ division.
It is interesting and practically useful to investigate if the sequences $C^{(n)}$ and $B^{(n)}$ eventually converge to a fixed pair of components, and if so, under what conditions.

\section{Related work}

\paragraph{Component-based design}
In~\cite{DBLP:conf/rtas/HenzingerM06} 
the authors present a calculus of time sensitive components that implement a set of sequences of time sensitive tasks. The model focuses more particularly on the timing profile of a task, and captures, in a component interface, the arrival time and latency for each sequence. Our work can benefit from~\cite{DBLP:conf/rtas/HenzingerM06} by using their formalism to specify and implement components. Alternatively, our family of interaction products on components may give new tools for compositional specification of real-time scheduling. 

In~\cite{DBLP:conf/amast/BidoitH08}, 
the authors present an axiomatic specification of components as pairs of an interface signature and a set of sentences capturing some pre and post conditions of stateful components.
The semantics of such an interface is the set of models that satisfy all axioms.
The authors present an implementation of such interfaces as component bodies, and also provide a composition operation that composes required and provided interfaces of two components. 
The difference with our work is mainly on the nature of the algebra: we model \emph{interaction} algebraically, while the authors model the interface, specification, and the body of a component algebraically. Both works can benefit from each other, as our work can be extended using an algebraic specification of components; and our interaction product can extend the algebra of components in~\cite{DBLP:conf/amast/BidoitH08}.

\paragraph{(De)composition} In~\cite{DBLP:conf/esop/ChenCJK12}, the authors present a declarative and an operational theory of components, for which they define a refinement relation and monotonicity results for some composition operators. Our work is related as it aims for similar results, but for the case of Cyber-Physical systems. Thus, instead of having input and output actions, our components have timed observations, and composability relations.
We present, as well, quotient operation on components, and show how it can be used to synthesize coordinating CPSs. 

In~\cite{DBLP:conf/facs2/PourvatanSAB10}, the authors consider the problem of decomposition of constraint automata. 
This work provides a semantic foundation to prove that the construction in~\cite{DBLP:conf/facs2/PourvatanSAB10} is a valid division.

\paragraph{Algebra, co-algebra}
The algebra of components described in this paper is an extension of~\cite{DBLP:journals/corr/abs-2110-02214}.
Algebra of communicating processes~\cite{DBLP:books/daglib/0000497} (ACP) achieves similar objectives as decoupling processes from their interaction. For instance, the encapsulation operator in process algebra is a unary operator that restricts which action occurs, i.e., $\delta_H(t \parallel s)$ prevents $t$ and $s$ to perform actions in $H$. Moreover, composition of actions is expressed using communication functions, i.e., $\gamma(a,b)=c$ means that actions $a$ and $b$, if performed together, form the new action $c$.
Different types of coordination over communicating processes are studied in~\cite{BERGSTRA1984109}.
In~\cite{DBLP:journals/jlp/BaetenM01}, the authors present an extension of ACP to include time sensitive processes. 

The modeling of a component's interaction using co-algebraic primitives is at the foundation of the Reo language~\cite{DBLP:conf/wadt/ArbabR02}. 
In~\cite{barbosa2001}, the question of separation of components into two sub-components is addressed from a co-algebraic perspective.

\paragraph{Discrete Event Systems}
Our work represents both cyber and physical aspects of systems with a unified model of discrete event systems.
In~\cite{doi:10.1146/annurev-control-053018-023659}, the author lists the current challenges in modelling cyber-physical systems in this way. 
The author points to the problem of modular control, where even though two modules run without problems in isolation, the same two modules may block when they are used in conjunction.
In~\cite{DBLP:journals/tac/SampathLT98}, the authors present procedures to synthesize supervisors that control a set of interacting processes and, in the case of failure, report a diagnosis. 
An application for large scale controller synthesis is given in~\cite{MOORMANN2021104902}.

\paragraph{Coordination}
In~\cite{Nivat1982}, the author describes infinite behaviors of process and their synchronization. Notably, the problem of non-blockingness is considered: if two processes eventually interact on some actions, how to make sure that they will not block each other.

\section{Conclusion}
We approach the challenge of designing cyber-physical systems using algebraic methods.
Components denote sequences of observations over time, and operations on components capture the interaction that arises from the behaviors of those components.
We extend a family of algebraic products with a corresponding family of division operators. 

We show how division can serve as a decomposition operator.
Intuitively, a component is divisible by another component if there exists a third component that, in composition with the latter gives the former. Division also opens some new reasoning possibilities as to update a system by replacing a component with a `better' component while preserving the overall behavior.
We apply our framework to reason about updates in a cyber-physical system consisting of robots moving on a shared field.

\paragraph{Acknowledgement}
Talcott was partially supported by the U. S. Office of Naval Research under award numbers N00014-15-1-2202 and N00014-20-1-2644, and NRL grant N0017317-1-G002.
Arbab was partially supported by the U. S. Office of Naval Research under award number N00014-20-1-2644.

\bibliographystyle{plain}
\bibliography{references}

\appendix
\section{Proofs}
\label{appendix:proof}

\begin{proof}[lemma~\ref{lemma:po-ref}]
Follows from reflexivity, antisymmetry, and transitivity of set inclusion. 
\qed
\end{proof}
\begin{proof}[lemma~\ref{lemma:po-cont}]
    Let $A = (E_A,L_A)$, $B = (E_B, L_B)$, and $C=(E_C,L_C)$ be three components.
    We show that $\leq$ is reflexive, transitive, and antisymmetric for any set $\mathcal{C}$ that satisfies the above condition:
    \begin{enumerate}
        \item reflexivity: $A \leq A$ holds.
        \item  transitivity. Let $A \leq B$ and $B \leq C$.
            Then, for all $\sigma:A$, there exists $\tau:B$ such that $\sigma\leq \tau$, and for all $\tau:B$, there exists $\delta:C$ such that $\tau \leq \delta$. Then, we conclude that for all $\sigma:A$, there exists $\delta:C$ such that $\sigma\leq \delta$ and $A \leq C$.
        \item antisymmetric. 
            We suppose that $A$ and $B$ are elements of the set $\mathcal{C}$.
            If $A \leq B$ and $B \leq A$, then for all $\sigma:B$, there exists $\tau:A$ such that $\sigma \leq \tau$. As well, for any $\tau:A$, there exists $\sigma:B$ such that $\tau \leq \sigma$.
            Thus, for any $\sigma:B$, there exists $\tau:A$ and $\delta:B$ with $\sigma \leq \tau \leq \delta$.
            Given the assumption of $A$ and $B$, we can conclude that $\sigma = \tau = \delta$.
            Similarly, we show that $L_A \subseteq L_B$, and that $A = B$.
    \end{enumerate}
    \qed
\end{proof}

\begin{proof}[Lemma~\ref{lemma:comp}]
Let $A$, $B$, and $C$ be three components, such that $B \sqsubseteq A$.
Then, the interface of $B \bowtie C$ is $E_B \cup E_C$, which is included in $E_A \cup E_C$ the interface of $A \bowtie C$.

For any TES $\sigma:B \bowtie C$, there exist two TESs $\beta:B$ and $\delta:C$ such that $(\beta,\delta)$ are synchronous, and $\sigma = \beta[\cup]\delta$. Since for any $\beta:B$ we also have $\beta:A$, then $\sigma$ is also an element of the behavior of $A\bowtie C$, and $B\bowtie C \sqsubseteq A \bowtie C$.
\qed
\end{proof}

\begin{proof}[Theorem~\ref{th:conf}]
Let $\mathcal{C}(E)$ be a finite subset of the set $\{ C \mid C$ \textit{ has interface  $E$ and } $C \in A\downarrow^*_{(R,\oplus)} B \}$.
    We define the union of two components $A = (E_A, L_A)$ and $B=(E_B, L_B)$, as the component $A \cup B = (E_A \cup E_B, L_A \cup L_B)$.
    The union of all components in $\mathcal{C}(E)$ is the component $\bigcup \mathcal{C}(E) = (E, \bigcup_{C \in \mathcal{C}(E)} L_C)$ where $L_C$ is the behavior of component $C$.
    Moreover, we have that, for any component $A, B, C$, with $B$ and $C$ sharing the same interface $E$, $(A \times_{(R,\oplus)} B) \cup (A \times_{(R,\oplus)} C) = A \times_{(R,\oplus)} (B \cup C)$. Indeed let $L$ be the behavior of $(A \times_{(R,\oplus)} B) \cup (A \times_{(R,\oplus)} C)$ and $S$ be the behavior of $A \times_{(R,\oplus)} (B \cup C)$:
    \begin{align*}
        L &= \{ \sigma \oplus \tau \mid \sigma \in L_A, \tau \in L_B, (\sigma, \tau) \in R(E_A, E) \} \cup \\
          &\quad \ \{ \sigma \oplus \tau \mid \sigma \in L_A, \tau \in L_C, (\sigma, \tau) \in R(E_A, E) \}\\
          &= \{ \sigma \oplus \tau \mid \sigma \in L_A, \tau \in L_B \cup L_C, (\sigma, \tau) \in R(E_A, E) \}
          = S
    \end{align*}

    We show that $\bigcup\mathcal{C}(E)$ is an upper bound for the set of coordinators $\mathcal{C}(E)$.
    For any $C \in \mathcal{C}(E)$, we have
    \[ C \sqsubseteq \bigcup \mathcal{C}(E)\]
    which implies that $C \leq \bigcup \mathcal{C}(E)$ and makes $\bigcup \mathcal{C}(E)$ an upper bound for $\mathcal{C}(E)$.

    Given associativity, commutativity, and idempotency of $\times_{(R,\oplus)}$, for any $C_1, C_2 \in \mathcal{C}(E)$:
    \begin{align*}
         B \times_{(R,\oplus)} C_1 & \sqsubseteq A \\
         B \times_{(R,\oplus)} C_2 & \sqsubseteq A \\
        (B \times_{(R,\oplus)} C_1) \cup (B \times_{(R,\oplus)} C_2) & \sqsubseteq A \\
        B \times_{(R,\oplus)} (C_1 \cup  C_2) & \sqsubseteq A 
    \end{align*}
    which, applied over the set $\mathcal{C}(E)$, gives $B \times_{(R,\oplus)} (\bigcup \mathcal{C}(E)) \sqsubseteq A$.
    Thus, $\bigcup \mathcal{C}(E) \in \mathcal{C}(E)$. 
    \qed
\end{proof}

\end{document}